\documentclass[10pt,a4paper,twocolumn, abstracton, listof=totoc, bibliography=totoc, toc=index, twoside, captions=oneline, headings=normal, parskip=half, headsepline, footsepline, numbers=nodotatend]{scrartcl}

\usepackage[latin1]{inputenc} 			%Damit ich ohne weiteres Umlaute schreiben kann
\usepackage{exscale} 				%Behebt Bug in Formelskalierung
\usepackage{amssymb}			%Mathesymbole
\usepackage{amsmath}			%Mathe sonst
\usepackage{amsthm} 				%Mathematische Umgebungen
\usepackage{amsfonts}			%Matheschriftarten
\usepackage[pdftex]{graphicx} 			%F�r Bilder
\usepackage[activate=normal]{pdfcprot} %Randausgleich bei PDFs
\usepackage{lmodern} %Neue Latexstandardschriftarten
\usepackage{ifthen}
\usepackage{xifthen}
\usepackage{calc}
\usepackage{cite}

\usepackage[pdftex,                            % sonst  gehts nicht
                bookmarks=true,% Lesezeichen-Leiste anzeigen
                bookmarksnumbered=true,        % Kapitelnummern in Lesezeichenleiste anzeigen
                bookmarksopen=true,
                bookmarksopenlevel=0,	%Alles bis Kapitel ist offen
                debug=false,
                hyperfootnotes=false,
                pdfauthor=David Kriesel,
                pdffitwindow=true,
                hyperindex,
                plainpages=false,
                pdfpagelabels]
                {hyperref} 				%Automatische Links in PDF

\newtheorem {theorem} {Theorem}

\newtheorem {lemma}   {Lemma}
\theoremstyle{definition}
\newtheorem {observation}   {Observation}
\newtheorem {corollary}   {Corollary}
\newtheorem {definition}    {Definition}

\setcapindent{0em} %Floatcaptionlabelsnicht ausger�ckt
\setkomafont{caption}{\sffamily\small} 
\setkomafont{captionlabel}{\sffamily\bfseries\small}

\newcommand{\boundingbox}[1]{\text{\textsl{box}}(#1)}

\newcommand{\poly}[1]{\text{\textsl{poly}}(#1)}%
\newcommand{\circum}[1]{\text{\textsl{circ}}(#1)}

\newcommand{\width}[1]{\text{\textsl{width}}(#1)}
\newcommand{\height}[1]{\text{\textsl{height}}(#1)}

\newcommand{\contaminationset}{\mathcal{C}}

\newcommand{\searchtime}[1]{\text{\textsl{searchtime}}(#1)}

\newcommand{\N}{\mathbb{N}}

\newcommand{\defi}[1]{\emph{#1}\defirand{#1}}
\newcommand{\defirand}[1]{\marginpar[
\begin{flushright}\footnotesize{\sffamily{#1}}\end{flushright}
]{
\footnotesize{\sffamily{#1}}
}}

\newcommand{\vect}[3][,]{%
{#2}_1 #1 %                                  bis zum ersten Trenner
{#2}_2 #1 %                                  bis zum zweiten Trenner
\ifthenelse{\equal{#1}{,}}{\dotsc}{\dotsb} % je nach Trenner benutze unterschiedliche Punkte
\ifthenelse{\equal{#3}{}}{}{#1  {#2}_{#3}}%  wenn der letzte Parameter nicht angegeben wurde, dann geht die Folge bis unendlich `a_1, a_2,...'
} %                                          macht `a_1, a_2,..., a_m'

\usepackage{xcolor} 				%Farben
\usepackage[plain]{algorithm} % Algorithmus floating Umgebung
\usepackage{listings}

\usepackage[capitalise]{cleveref}
\Crefname{lemma}{Lemma}{Lemmata}

\definecolor{pseudocodebg}{rgb}{0.854,0.894,0.969}
\definecolor{pseudocodekeyword}{rgb}{0.262,0.329,0.439}
\definecolor{pseudocodestring}{rgb}{.6,0,0} %darkred
\definecolor{pseudocodecomment}{rgb}{.4,.4,.4} %grey

\lstdefinelanguage{pseudo}
	{
		morekeywords={IF,THEN,ELSE,FOR,ALL,WHILE,DO,UNTIL,END,REPEAT,AND,OR,NOT,RETURN,ELSE IF,SUBROUTINE,ALGORITHM},
		sensitive=false,
		morecomment=[l]{\#},
		morestring=[b]",
	}
\title{A local strategy for cleaning expanding cellular domains by simple robots}
\author{Rolf Klein, David Kriesel, Elmar Langetepe\\\normalsize{Department of Computer Science 1}\\\normalsize{University of Bonn}\\\normalsize{\texttt{rolf.klein@uni-bonn.de}, \texttt{mail@dkriesel.com}, \texttt{elmar.langetepe@cs.uni-bonn.de}}}
\date{\today}

\newlength{\myalgowidth}
\newlength{\myalgonumberingwidth}
\begin{document}
\maketitle

\begin{abstract}
We present a strategy {\em SEP} for finite state machines tasked with cleaning a cellular environment in which a contamination spreads. Initially, the contaminated area is of height $h$ and width $w$. It may be bounded by four monotonic chains, and contain rectangular holes. The robot does not know the initial contamination, sensing only the eight cells in its neighborhood. It moves from cell to cell, $d$ times faster than the contamination spreads, and is able to clean its current cell. A speed of $d<\sqrt{2}(h+w)$ is in general not sufficient to contain the contamination. Our strategy {\em SEP} succeeds if $d \geq 3(h+w)$ holds. It ensures that the contaminated cells stay connected. Greedy strategies violating this principle need speed at least $d \geq 4(h+w)$; all bounds are up to small additive constants. 
\end{abstract}


\section{Introduction}
\label{s_introduction}

%\TODO{Make up for, oder sowas wie disturbed, muss raus}

%\TODO{Lemma 12: Das was cleand ist mehr als das was gem�� Lemma xx dazukommt}

%\TODO{alle disturbed raus, "if no spread occurs"}

%\TODO{wichtige Definitionen in Definitionsumgebungen}

%\TODO{Traversal neu definieren, darin enthalten lassen dass kein spread vorkommt}

%\TODO{To be most precise, etc: raus}

%\TODO{Lemma9 Simplyconnected!!}

%\TODO{Ohrendefinition aus den Beweisen raus, Bild hinzuf�gen, im Lemma mit den 9 Ohrentypen sagen, dass das Ohrentypen sind!}

%\TODO{Alle Definitionen, die in Beweisen sind, raus.}

%\TODO{�berschrift Quicksearch: "More efficient Boundary Search"}
%\TODO{Auch in QS sagen, dass nur die Suchzeit reduziert wird und das Verhalten w�hrend der Cleaningphase nicht beeintr�chtigt wird, Schwammigen Absatz in Section 6 suchen und raus}

%\TODO{�berall: Two more left than right turns etc. raus}

%\TODO{"left", "west" usw vereinheitlichen}

%\TODO{Figure 14: Auf Unterf�lle hinweisen, Fall 6 sind eigentlich 4}
%\TODO{Figure 14: Quadranten auszeichnen in denen ein Freie Zelle sein muss, sonst Ursprungskontamination nicht in C}

During the last years, researchers in technical fields have become increasingly fascinated by the potential of biological, decentrally organized systems. From myriads of fireflies powdering entire meadows with shallow light, even flashing in a synchronized way, to ant colonies with sometimes millions of individual beings building most sophisticated structures that allow for air conditioning, storage and even growth of food, such systems exhibit fault-resistance and cost-efficiency while being flexible and able to solve most complex tasks (an extensive survey can be found at, for instance, \cite{garnier2007biological}). 

Understanding such phenomena represents a serious challenge to theoretical computer science. Although there is a rich body of work on autonomous agents, comparably few papers offer theoretical results on agents who have limited perception, limited computing and  translocating capabilities, and yet successfully deal with dynamically changing environments.

 In this paper we are studying cellular environments in the plane. Two cells are adjacent if they share an edge. At each time, finitely many cells may be contaminated, all others are clean. 
 
\begin{definition}
The set of all contaminated cells at a time is called \defi{contamination}. 
\end{definition}

\begin{definition}\label{def_contaminationset}
We assume that an initial contamination $C$ has the following geometric properties. It is connected, and its outer boundary consists of four monotonic chains; they connect the extreme edges supporting the bounding box of $C$. Inside, $C$ may contain rectangular holes consisting of clean cells; see \cref{p_initialcontamination}. Let $\contaminationset$\defirand{$\contaminationset$} be the set of all such contaminations. 
\end{definition}

\begin{definition} Inspired by forest fires or oil spills, every \defi{$d$} time units a contamination spreads from each contaminated cell to its four neighbors, as shown in \cref{p_initialcontamination-spread}.
\end{definition}
\begin{figure}
\centering
\includegraphics[width=\columnwidth]{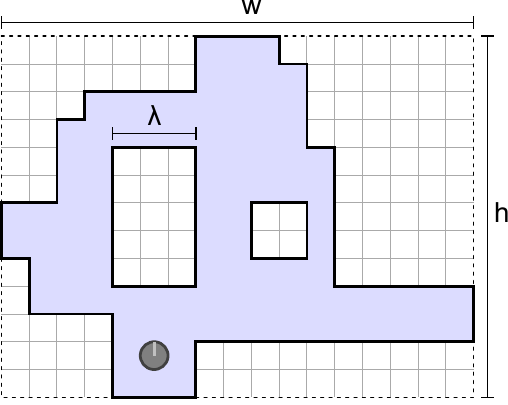} 
\caption{An initial contamination in $\contaminationset$. Also depicted are the contamination's axis aligned bounding box (the rectangle outlined in a dashed way), its width ($w$) and its height ($h$). $\lambda$ is the length of the longest short side among any holes.}
\label{p_initialcontamination}
\end{figure}
\begin{figure}
\centering
\includegraphics[width=\columnwidth]{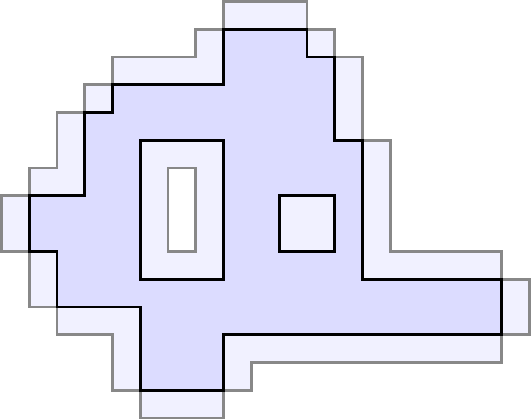} 
\caption{The contamination from \cref{p_initialcontamination}. The cells that will  become contaminated during a spread are colored lighter.}
\label{p_initialcontamination-spread}
\end{figure}

We want to enable a robot to clean the contamination. Initially, the robot is located in one of the contaminated cells. It can sense the status of the eight cells in its neighborhood; see \cref{p_initialcontamination}. In each time unit, the robot can turn, move to one of the four adjacent cells, and decide to clean it. Thus, $d$ measures the robot's speed against the contamination's. The robot is a finite automaton. It has no previous knowledge about $C$, and because its memory is of constant size it cannot store a lot of information as it moves around. There is no global control or any other information the robot could make use of.

Whether or not environment $C$ can be cleaned depends on its initial extension and its spreading time relative to the robot's speed, $d$. Let $h$ and $w$ denote height and width of the bounding box of $C$, respectively. In our model the perimeter of $C$, i.e., the number of edges on its outer boundary, equals $2(h+w)$. Thus, $h+w$ is a reasonable measure for the size of $C$; see \cref{p_initialcontamination}.

In Theorem~\ref{lowbound-theo} we will establish a geometric lower bound in terms of $h$ and $w$.
No robot can clean all environments of height $h$ and width $w$ if its speed $d$ is less than $\sqrt{2} (h+w)-4$ (not even if the robot knows $C$ and has Turing machine power).

Our main contribution is a strategy {\em Smart Edge Peeling (SEP)} for which we can prove the following performance guarantee.  Let $\lambda$ denote the maximum length of all shorter edges of the rectilinear holes inside $C$; see \cref{p_initialcontamination}.

\textbf{Theorem 1.} \emph{
Given speed $d \geq 3(h+w) +6$, and starting from a contaminated cell, strategy \emph{SEP} cleans each contamination in $\contaminationset$ of height $h$ and width $w$ in at most $(\frac{\lambda}{2} + h + w +5) d$ many steps.}

Starting from a contaminated cell, strategy {\em SEP} heads for the outer boundary of $C$, without attempting to enlarge any holes. Then it carefully peels the perimeter of $C$, layer by layer, making sure that the set of contaminated cells always stays connected. In order to maintain this invariant the strategy will {\em not} clean \defi{critical} contaminated cells which would destroy connectivity locally. The strategy is precisely defined  in \cref{s_strategy}. We have also implemented the strategy.  
A supplementary video of an execution of the strategy can be found at
\url{http://tizian.informatik.uni-bonn.de/Video/smartedgepeeling.mp4}\,.

\begin{definition}
Let $c$ be a contaminated cell. Let $S$ be the set of $c$ and its eight neighbors. Then, $c$ is considered critical if there exist contaminated cells $a, b$ in $S$ with $a \neq b \neq c$, so that all contaminated paths from $a$ to $b$ necessarily lead through $c$.
\end{definition}

This concept is taken from pixel-space filling algorithms in computer graphics \cite{henrich}; see \cref{p_criticality} for an example. Although this constraint causes extra cell visits and possible delay (when cleaning can only be continued once a spread has occurred), our strategy compares favorably with a greedy approach that does not care about connectivity.
\begin{figure}%
\center
\includegraphics[width=\columnwidth]{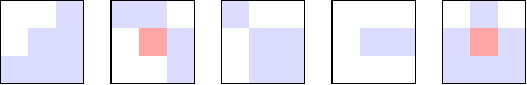}%
\caption{In the five examples depicted, the second and the fifth one's center cell is critical. With respect to the third example, note that we do not consider diagonally adjacent contaminated cells as connected.}%
\label{p_criticality}%
\end{figure}

In fact, a greedy strategy may completely clean one connected component, being unaware of others. Only after many spreads would it sense that contamination has again reached the robot's current position. In \cref{neogreedy-theo} we show that this approach can fail if speed $d$ is less than $4(h+w)$, whereas {\em SEP} always works if $d \geq 3(h+w) +6$, by \cref{main-theo}. Maintaining connectivity has the additional advantage that the robot knows when the very last cell has been cleaned, so that it could turn to another task.

The rest of this paper is organized as follows. In \cref{s_previous} we review related previous work.  Strategy {\em SmartEdgePeeling} is presented in \cref{s_strategy}. We show how to take advantage of the geometry of the scene (rectangular holes, an outer boundary consisting of four monotonic chains), and design {\em SEP} in such a way that these properties are also maintained under spreads and cleaning activities; proving these invariants is a major part of our analysis (\cref{s_spreadshape,s_strategychanges,s_quicksearch}). Our  main theorem introduced above is then proven in \cref{s_correctnessruntime}. \cref{s_lowerbounds} contains the lower bounds mentioned above, and in \cref{s_conclusions} we state questions for future work.

\section{Previous Work} 
\label{s_previous}
The problem of cleaning expanding domains is located within the field of \defi{robot motion planning}, which can itself be divided into several sub-fields dependent on the robot's computational capacities, the a priori knowledge it is given, and the kind of environment it finds itself in. 
%In this section, we will locate our work within this landscape and provide examples of related work for an overview.

In \defi{offline motion planning}, robots are in possession of all relevant information about the problem instance to solve, allowing them to plan their actions in advance, usually employing powerful computational capacities. A good example for this is the family of pursuit-evasion problems (for an overview, see \cite{fomin2008annotated}), also known as intruder search,  cops and robbers, or lion and man problems \cite{lionsdumitrescu, lionsberger}. In these problems, the space of the intruder's possible positions expands and needs to be contained, which can be seen as a rough analogon to the contamination in this article.
%Here, searchers have to apprehend an agile and omniscient fugitive within a certain environment, for example a graph or polygonal terrain.  Depending on the searchers' actions, the area of the fugitive's possible positions changes dynamically. Even though the behavior and clearing of this area can be seen as a rough analogon to cleaning expanding contaminations, agents in pursuit-evasion problems usually are supposed to have Turing machine-like computational capabilities and offline knowledge.

In \defi{online motion planning} scenarios, robots have to collect environment information at run time, for example by local sensor-based perceptions. Our scenario can clearly be located within this field, which however mostly considers static scenarios. A further field somewhat related to our work is online graph exploration (see, for example \cite{fomin2008annotated}), as our robots essentially explore dynamic grid graphs. However, in  graph exploration problems, in general, less use of geometrical properties is made. 

%Finer distinctions with respect to online motion planning scenarios can be made by what \defi{global information} robots are given during run time in addition to their own, local perceptions, and also by computational capabilities robots employ to accumulate knowledge.
%Global information can be provided by some externally maintained data structure robots can access, imagine for example a live radar view of relative positions of opponent robots or similar.
%Computational capabilities are important as they allow for \defi{knowledge accumulation}: Even if a robot has no a priori knowledge of its environment, with sufficient memory and computational power, it could build up something like a map containing all environment characteristics important to solve its task. 

%\subsection{Covering with mobile robots}

More precisely, our problem can be located in the range of \defi{mobile robot covering} problems, a form of terrain exploration requiring a robot to visit ("cover") all places within a given planar terrain. 
%This form of exploration can be motivated by short range sensorics, forcing the robot  to visit the whole terrain even for mapping purposes, but also by service robotics: For tasks like cleaning, de-mining or lawn mowing it is actually necessary to visit every place of the area in question. 
There are two common approaches to covering: \defirand{heuristic}Heuristic (for a recent example employing the common A* algorithm see\cite{6269300}) and \defirand{analytical}analytical, the latter of which aims for guaranteeing complete coverage. Finding optimal-lenght covering paths (also referred to as "lawn mower problem") is proven to be NP-Hard \cite{lawnmowingmilling} which leads to finding approximate solutions. In his 2001 survey \cite{choset2001coverage}, Choset classifies analytical coverage  approaches with respect to environments tesselated into a grid of square, close-to-robot-sized cells as \defi{approximate cellular decompositions}; see also \cite{galceran2013survey} for a more recent survey. 
%Some strategies employ \defi{exact cellular decompositions} of the terrain to cover, like the boustrophedon decomposition \refmark, which usually requires offline knowledge or at least long range sensors, as well as computationally powerful robots. 

%In contrast, strategies using \defi{approximate cellular decompositions} recognize the terrain as tesselated into a grid of close-to-robot-sized cells, usually of square shape. 
To this category our strategy presented in this article can be counted. There exist offline approaches \cite{zelinsky1993planning} as well as online ones, e.g.  \cite{elmar,gabriely2002spiral,gabriely,gonzalez2005bsa}. There are also bio-inspired ways of covering making use of stigmergic information like for example exhibited by ants \cite{wagner1999distributed}.

To the best of our knowledge, we present the first online approach on cleaning expanding grid domains, without adding global information or accumulating knowledge. In addition, in particular when approached in an analytical way, covering is usually adressed with respect to static environments. In contrast, cleaning expanding domains can be seen as a dynamic variant of mobile robot covering.

%TODO {\sc Grid cleaning: include lawn mowing, snow removal, milling; exploration, covering} <- gute links in spiralSTC

%TODO {\sc Fekete!}

%TODO {\sc Intruder search: usually more computing power and global knowledge}

\subsection{Cleaning static and expanding grid domains}

Cleaning of grid domains has been investigated first in both a static and a dynamic variant in a family of articles that serve as inspiration for our work \cite{ altshuler2005swarm, altshulershape, altshulerstatic, altshulerdynamic, altshulercomplexities}. 
%In these articles, not only the problems of cleaning both static and expending domains is introduced, they also represent the state of the art with respect to this problem, and we will extensively discuss them in the course this subsection.

In \cite{altshulerstatic} the special case of static contaminations without holes is addressed. The authors let robots traverse the boundary without central control, peeling off layers by cleaning any non-critical cell they encounter ("edge peeling"). 
%Their strategy results in an upper bound on the cleaning time $\in \gO(xy)$ for one single robot, where $x$ is the number of boundary cells and $y$ is the largest manhattan distance to a clean cell among all contaminated ones. The analysis is made for $k$ robots.

The authors propose to reuse their strategy with slight variations on the dynamic variant of the problem and present upper and lower bounds on the cleaning time \cite{altshulerdynamic, altshulercomplexities, altshuler2005swarm}. Extending strategies from the static problem to its dynamic version turns out to be difficult: In particular one cannot neglect holes even if an initial contamination is required to be simply-connected because at any spread parts of the contamination boundary may grow together and create holes out of former boundary parts. Holes however impose serious challenges: 
%For instance, robots need to make sure not to enlarge holes by cleaning operations, but to clean cells at the outer contamination boundary before it grows exceedingly. 
Distinction between the outer contamination boundary and hole boundaries with local knowledge (\cref{p_views}) is non-trivial. It needs to be ensured that robots do not get stuck at holes while the outer contamination boundary expands exceedingly.
%in general given that there are no restrictions on the contaminations size and the paradigm of finite state machines does not allow advanced means of navigation like map building.  

%\TODO{Der satz muss sauber ersetzt sein bevor er rauskommt} Further, holes may cause critical cells in the boundary of contaminations, forcing edge-peeling agents to create arbitrarily complex contamination structures. Eventually, robots and spreads may even complement one another in creating new, arbitrarily complicated holes (\cref{p_complexgeometry}).

\begin{figure}%
\center
\includegraphics[width=5cm]{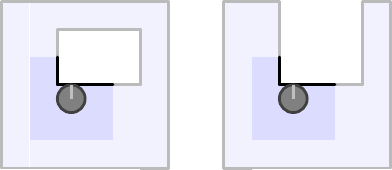}%
\caption{Even though the left robot sees a hole boundary, while the right one sees the outer contamination boundary, both robots have exactly the same perceptions.}%
\label{p_views}%
\end{figure}

In \cite{altshulerdynamic}, the authors disabled the existence of holes by not allowing them in initial contaminations and requiring the existence of a very helpful \emph{elastic membrane}.  The membrane is an automatically updated global data structure in the grid world that all agents share, see \cref{p_membrane}.  The robots traverse the membrane instead of the boundary and everything else basically stays the same. On the pro side, it enforces the contamination's simply-connectedness over time. On the other hand, the contamination's geometry can get arbitrarily complicated, see \cref{p_membrane}. In our work, we allow holes and only require an initial contamination's geometry to be relatively simple. We guarantee to maintain this simplicity and as a consequence, that no further holes emerge. We do not need a global data structure. The initial geometrical simplicity we require would not have helped in the related work at all, as in the original edge peeling strategies, such a simplicity is not maintained, so new holes can still emerge any time, see \cref{p_complexgeometry}.

\begin{figure}%
\center
\includegraphics[width=\columnwidth]{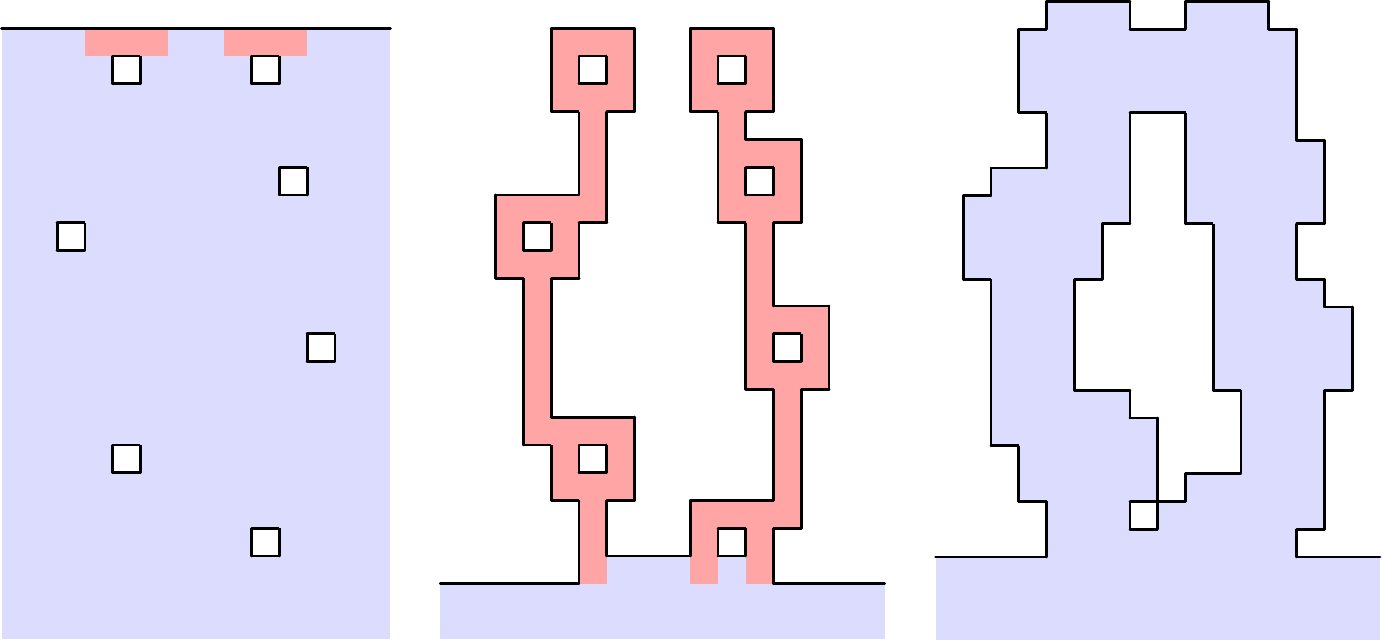}%
\caption{In the left image part, a section of a large, solid rectangular contamination with some minimal holes is illustrated. In the middle, the same section is illustrated after several "edge peeling" robot traversals. On the right side, we can see the spread outcome of the configuration in the middle.}%
\label{p_complexgeometry}%
\end{figure}

\begin{figure}%
\center
\includegraphics[width=\columnwidth]{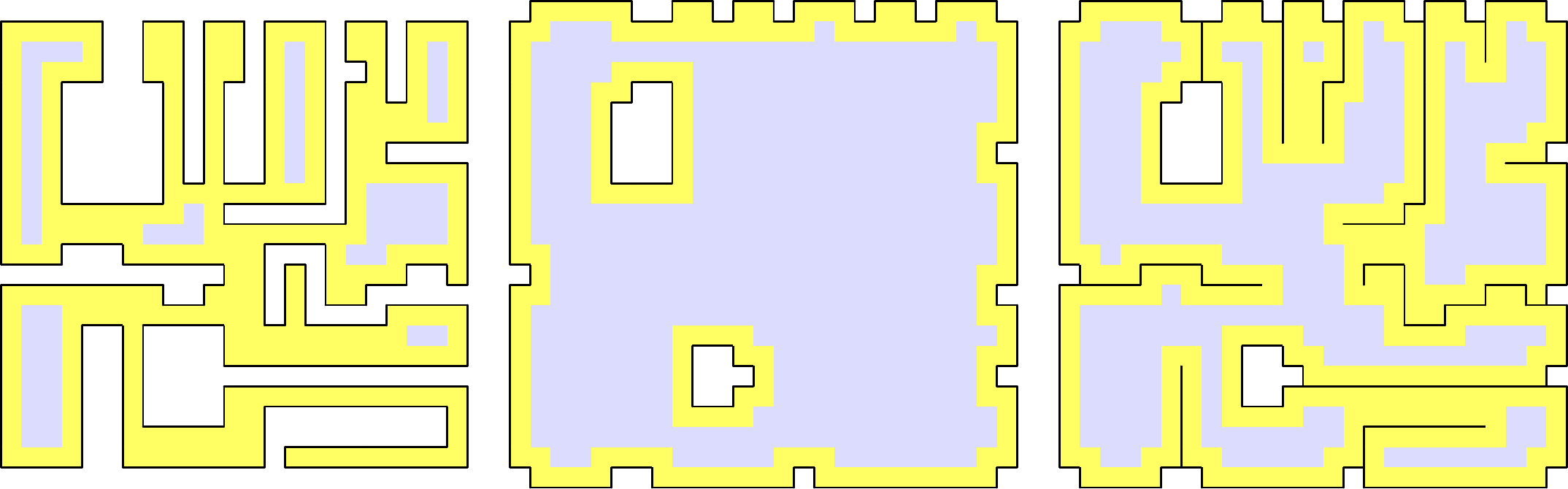}%
\caption{A contamination $C$ (left) and the spread outcomes of $C$ both without membrane (middle) and with membrane (right). Observe that, even though holes emerge, the contamination cannot grow together and stays simply-connected, if the elastic membrane restriction is placed. The membrane does not allow closed contaminated paths around emerging holes.}%
\label{p_membrane}%
\end{figure}
%Note that the outer boundary of contaminations in $\contaminationset$ cannot two left turns cannot contain a subsequence with two left turns, but no right turn (referred to as \defi{U-turn} in the following).

%Unfortunately in combination with spreads and the robot's cleaning, the membrane can become arbitrarily meandering and thus heavily alleviate cleaning efficiency. No guarantees on the contaminations' shape over time can be made. 
%In  a further article \cite{altshulercomplexities}, the authors improve their upper bound on the static problem to a bound $\in \gO(n^{1.5}+n)$ (put for $k=1$ robots here) with contaminations of $n$ cells. They also derive an upper bound in $\gO(n^2 \ln n)$ for cleaning membrane-restricted dynamic contaminations of initial size $n$ using $\Theta(\sqrt{n})$ robots with restrictions on the spreading time $d$. 

The authors mention that their strategy could work without a membrane if agents do not get stuck at holes like described above; In this work, we feel it necessary to ensure by only local means that this does never happen.

\section{How spreads change contaminations}
\label{s_spreadshape}

In order to prove that a contamination does not change the complexity of its geometry during a spread, we firstly need to define some fundamentals.
Let $C$ be a contamination and $a,b$ adjacent cells with $a\in C$ and $b\notin C$. 

\begin{definition}
The edge separating $a$ and $b$ is called a \defi{border edge} and $C$ is enclosed by a simple grid polygon, $\poly{C}$. 
\end{definition}
\begin{definition}
With respect to a clockwise traversal, grid polygons can be seen as an intersection-free closed sequence of \defi{atoms}, where atoms can be of type \emph{border edge}, \defi{right turn} and \defi{left turn}, see \cref{p_atoms}. Note that no turns are located next to each other. The \defi{length} of any such sequence is the number of border edges it contains.
\end{definition}

We use the following naming conventions. $\boundingbox{C}$ denotes the axis aligned \defi{bounding box} of $C$, and $\width{C}$ and $\height{C}$ the extension from west to east and north to south respectively, see \cref{p_initialcontamination}. 

\begin{figure}%
\center
\includegraphics[width=\columnwidth]{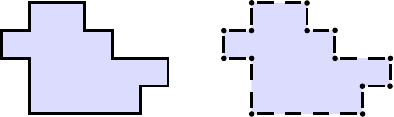}%
\caption{A contamination depicted as enclosed by a grid polygon on the left, and with the grid polygon as a sequence of atoms on the right. On the right, the line segments mark border edges, and the dots mark left and right turns (with respect to a clockwise traversal).}%
\label{p_atoms}%
\end{figure}

\begin{lemma}\label{lem1_compltransition_spread}
Let $C \in \contaminationset$. Let $D$ be the outcome of $C$ after a spread. Then, $D \in \contaminationset$.
\end{lemma}
\begin{proof}
Contaminations in $\contaminationset$ are enclosed by four monotonic chains. No parts of $\poly{C}$ can grow together by a spread, as such configurations would require $\poly{C}$ parts growing towards each other and therefore  $\poly{C}$ to contain an U-turn. Therefore, no new holes can emerge and, as contaminated cells only contaminate their 4Neighborhood, $\poly{D}$ consists of four monotonic chains again.

After a spread, the rectangular holes have disappeared or are smaller and still rectangular. Furthermore, as $C$ contained finitely many contaminated cells, $D$ also does. Last, as $C$ was connected, $D$ is also connected. Hence, $D \in \contaminationset$.
\end{proof}

%It remains to examine a spread's influence on a contamination's width, height and circumference. 
\begin{definition}
The \defi{circumference} of a contamination $C$, denoted as $\circum{C}$, is the length of the shortest closed path of cells $\in C$ that touches every border edge in $\poly{C}$.
\end{definition}

Note that dependent on $\poly{C}$'s shape, some cells may be visited more than once, see \cref{p_circumference}.
Because of the monotony of the four parts $\poly{C}$ consists of, for  $C \in \contaminationset$, we know:

\begin{lemma}\label{lem_circumferencelenght}
Let $C \in \contaminationset$. Then, $\circum{c} = 2\width{c} + 2\height{c}-4$. 
\end{lemma}

Also, obviously, the following holds:
\begin{observation}\label{observation_spread}
Let $C$ be a contamination. Let $D$ be its spread outcome. Then $\height{D} = \height{C} + 2$ and $\width{D} = \width{C} + 2$.
\end{observation}
%\end{lemma}

\begin{figure}%
\center
\includegraphics[width=4cm]{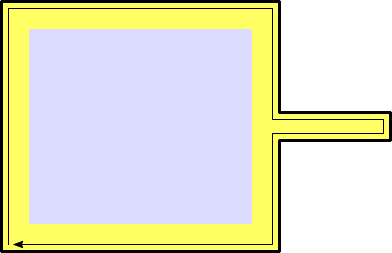}%
\caption{The circumference of a contamination $C$ is the number of steps an agent would have to perform in order to traverse the cells touching $\poly{C}$ completely.}%
\label{p_circumference}
\end{figure}

 From this, we also know  $\circum{D} = \circum{C} + 8$.
 Also note that by spreads, holes are retracted from the outer contamination boundary. In order to formalize this, we need another definition:
 \begin{definition}
Let the set of all cells touching $\boundingbox{C}$ be denoted as \defi{layer} 1. Let the set of 4Neighbors of layer 1 located to the inner side of layer 1 be denoted as layer 2, and so on until every cell in $\boundingbox{C}$  is assigned to a layer, see \cref{p_layers}.
\end{definition}

\begin{figure}%
\center
\includegraphics[width=55mm]{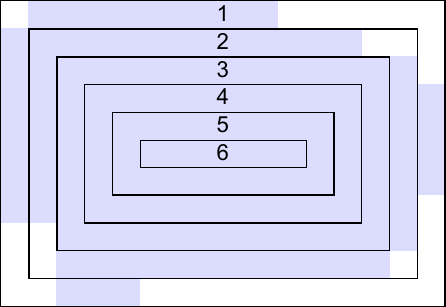}%
\caption{A contamination partitioned into layers with layer numbers depicted.}%
\label{p_layers}%
\end{figure}

\begin{lemma}\label{lem_4layers}
Let $C$ be a contamination and let $C$ contain holes. Let $D$ be the spread outcome of $C$. Then, the outmost hole cell in $D$ is located in a layer $\geq 4$, if such exists.
\end{lemma}
\begin{proof}
As holes are surrounded by closed paths of contaminated cells, the outmost hole cell in $C$ is in layer $\geq 2$. Layer 2 with respect to $C$ will be layer 3 with respect to $D$. Additionally, all cells in Layer 2 with respect to $C$ have a contaminated 4Neighbor. Thus, all hole cells in layer 2 with respect to $C$ get contaminated during the spread and the outmost remaining hole cell can only be in a layer $\geq 3$ with respect to $C$ or $\geq 4$ with respect to $D$.
\end{proof} 
 
 %In summary, we so far investigated how contaminations are changed by spreads, so to say, how the world changes itself. We can now proceed defining cleaning strategy SEP, and then investigate how contaminations are changed by an agent using this strategy. 
 
% \TODO{Vereinheitlichen dass SEP manchmal hervorgehoben ist und manchmal nicht}

\section{Cleaning strategy: Smart Edge Peeling}
\label{s_strategy}

%TODO {Braitenbergreferenz in Herrn Kleins Roboterbeschreibung an den Anfang}

The agent keeps an integer bearing counter for the turns performed, which is initialized with 0. It is increased at $90^\circ$ right turns and decreased at $90^\circ$ left turns. Note, that the bearing counter can only represent three values, so the limitations of the finite automaton robot model are not violated. Based on its perception of the eight cells around its position, an agent can decide whether or not its position is critical. 

We will now define cleaning strategy \emph{SEP}. In contradistinction to the related work, it will not clean every uncritical boundary cell it encounters to guarantee that the contamination's geometrical invariants are preserved. 
A supplementary video of an execution of the strategy can be found at
\url{http://tizian.informatik.uni-bonn.de/Video/smartedgepeeling.mp4}\,.

We assume that in any time step, first a spread takes place if $t = nd$ for $n \in \N$, and second the strategy is started (in the first time step) or resumed (in any further time step). Further, $d > 1$, so there cannot be a spread directly before the first time step. The strategy \defi{Smart Edge Peeling (SEP)} is presented formally in \cref{alg_smartedgepeeling}. The comments are meant to be read together with the below explanations. 
\lstset{escapeinside={(*}{*)}}
\begin{algorithm*}[tbp]
\caption{Strategy \emph{Smart Edge Peeling (SEP)}. Precondition: Agent starts on a contaminated cell. For the sake of readability, it is assumed that the status variables \texttt{bearingCounter} (incremented for each right turn and decremented for each left one) and \texttt{lastTurnWasRight} are maintained automatically. The status of  \texttt{bearingCounter} is assumed to be only updated in search mode.}
\label{alg_smartedgepeeling}
\begin{tabular}{p{\myalgowidth}p{\myalgonumberingwidth}p{\myalgowidth}}
\begin{lstlisting}
IF spreadDetected OR firstTimeStep: (*\label{ln_spreaddetected}*)
	#Initialize / reset status variables
	lastTurnWasRight = "false"; #Autom. maintained
	bearingCounter = 0; #Autom. maintained
	mode = "search"; #Can be "search" or "boundary"
	criticalCellPassed = "true";
	
IF locatedOnLastContaminatedCell:
	cleanCurrentLocation;
	terminate;

#BLOCK 1. Makes the agent turn right (*\label{ln_block1}*)
#and assesses if it has reached the boundary.
IF mode == "search" AND bearingCounter == 0: (*\label{ln_search}*)
	#Agent moves freely through the contamination. It
	#may be able to continue its movement or forced 
	#to turn by clean cells around it.
	IF frontContaminated: 
		#Contaminated cell ahead. No turn needed.
	ELSE IF rightContaminated: 
		#Clean cell ahead, contaminated one on the right.
		turnRight; (*\label{ln_search_turnright}*)
	ELSE IF rearContaminated:
		#Front and right side are clean; turn right two times. 
		#BC reaches 2, which cannot be caused by holes.
		#So, the agent switches to boundary mode and 
		#resets criticalCellPassed. This flag may be set true 
		#later dependent on the current position's criticality.
		turnRight; turnRight; mode = "boundary"; (*\label{ln_switchtoboundary1}*)
		criticalCellPassed = "false";
	ELSE IF leftContaminated:
		#Front, right and rear are clean. Turn right three  
		#times and analogously switch to boundary mode.
		turnRight; turnRight; turnRight; (*\label{ln_switchtoboundary2}*)
		mode = "boundary"; criticalCellPassed = "false";
ELSE IF mode == "search" AND bearingCounter == 1:
	#Agent follows a boundary to its left, but it cannot
	#know if the boundary belongs to a hole or not.
	IF leftContaminated: (*\label{ln_ifline}*)
		#Turn back to original orientation with BC 0.
		turnLeft; (*\label{ln_search_turnbackleft}*)
\end{lstlisting}
&\  &
\begin{lstlisting}
#BLOCK 1 continued. Be aware that the next line is an
#ELSE IF with respect to the IF in line (*\color{pseudocodecomment}{\textsl{\cref{ln_ifline}}}*).
	ELSE IF frontContaminated:
		#Left side still clean, contaminated cell ahead.
		#No turn needed, continue to follow boundary.
	ELSE IF rightContaminated:
		#Clean cells ahead and to the left. Agent turns
		#right again, knowing to have reached the boundary.
		turnRight; mode = "boundary"; (*\label{ln_switchtoboundary3}*)
		criticalCellPassed = "false";
	ELSE IF rearContaminated: 
		#Clean cells to the left, front and right. Agent
		#knows to have reached boundary, turns right twice.
		turnRight; turnRight; mode = "boundary"; (*\label{ln_switchtoboundary4}*)
		criticalCellPassed = "false";
ELSE IF mode == "boundary": (*\label{ln_boundary}*)
	#Agent knows it is at the boundary, follows it using 
	#left hand rule. After right turns, it resets 
	#criticalCellPassed (again: may be set true later).
	IF leftContaminated: 
		turnLeft; (*\label{ln_boundary_left}*)
	ELSE IF frontContaminated:
		#Agent just goes ahead, no turn needed (*\label{ln_boundary_front}*)
	ELSE IF rightContaminated:
		turnRight; criticalCellPassed = "false"; (*\label{ln_boundary_right}*)
	ELSE IF rearContaminated: 
		turnRight; turnRight; criticalCellPassed = "false"; (*\label{ln_boundary_rear}*) (*\label{ln_block1end}*)

#BLOCK 2: Cleaning and Movement. The agent is (*\label{ln_block2}*)
#allowed to clean if it is in boundary mode, the
#last turn was a right one and no critical cell 
#was passed after this right turn.
IF currentPositionCritical: 
		criticalCellPassed = "true";	(*\label{ln_critical_passed}*)
IF mode == "boundary": 
	IF lastTurnWasRight AND NOT criticalCellPassed:
		cleanCurrentLocation; (*\label{ln_cleaning_normal}*)
	ELSE IF inTailCell:
		#Tail cells are always cleaned when in 
		#boundary mode.
		cleanCurrentLocation; (*\label{ln_cleaning_tail}*)
moveForward;(*\label{ln_block2end}*)
\end{lstlisting}
\end{tabular}

\end{algorithm*}

Let the agent be deployed somewhere within a contamination $C \in \contaminationset$. W.l.o.g. let the agent be starting  northwards.

In \defi{search mode} (from \cref{ln_search} on), the agent does not clean any cell, but moves through the contamination searching for the outer boundary. If it encounters any boundary, it turns right (\cref{ln_search_turnright}), increasing its bearing counter. It cannot know if it has found a hole or the outer contamination boundary. As holes are rectangular, they will force the agent to perform a right turn. However, the agent will turn left and move northwards again later on, decreasing the bearing counter to 0 again (\cref{ln_search_turnbackleft}). In this manner, the agent will leave any hole encountered without doing any change to it, (\cref{p_lochausweichen}) eventually encountering the outer contamination boundary. Once the bearing counter reaches 2, the agent will know for sure to have reached the outer contamination boundary, and switch to \defi{boundary mode} (from \cref{ln_boundary} on). Let  $\searchtime{C}$ be the time the agent has spent until switching to boundary mode.

\begin{lemma} \label{lem_searchtime}
Let $C \in \contaminationset$. Then, $\searchtime{C} \leq \width{C} + \height{C}-2$.
\end{lemma}
\begin{proof}
Follows from the fact that the agent has been moving monotonously towards the east and the north (\cref{p_lochausweichen}).
\end{proof}

\begin{figure}%
\center
\includegraphics[width=\columnwidth]{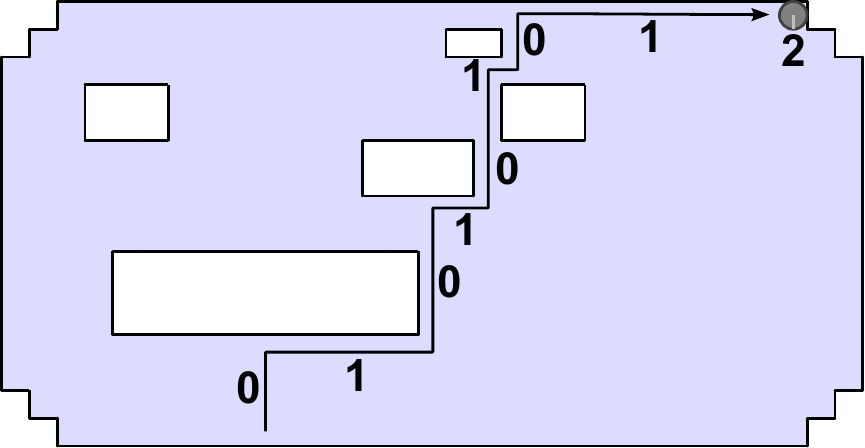}%
\caption{An example search mode trajectory that initially performed by an agent using \emph{SEP}. As the agent encounters several holes, it's bearing counter changes between 0 and 1. The first time the agent's bearing counter reaches 2, the agent switches to boundary mode.}%
\label{p_lochausweichen}%
\end{figure}

In boundary mode, the agent will follow the boundary using left-hand rule.   Before we describe the boundary mode in detail, we need  further definitions:
\begin{definition}\label{def_traversal}
Let $C \in \contaminationset$ and an agent be in boundary mode in $C$. Then we call the following $\circum{C}$ steps the agent performs one \defi{traversal}, if no spread occurs within this time period.
\end{definition}

%\begin{definition}
%We call boundary cells that are not cleaned because of a left turn or a critical cell that was traversed beforehand \defi{uncleanable}. \TODO{Diese Definition ist mir noch etwas schwammig ...}
%\end{definition}

\begin{definition}
We call a contaminated cell that touches at least three border edges a  \defi{tail}.
\end{definition}

We have to use the time dependence on $C$ for this definition, as we cannot easily put a cell-based traversal definition, for the agent might clean cells and reduce the circumference during a traversal.
%Analogously we may call, for example, $\frac{\circum{C}}{2}$ steps in this way a half traversal, and so on.

Only in boundary mode, cells are cleaned, and critical cells are omitted from cleaning. This immediately leads the following lemma:
\begin{lemma}
\emph{SEP} does not destroy a contamination's connectivity.
\end{lemma}

However, cleaning is controlled even more carefully (see \cref{p_cleanstartstop}): The agent starts a cleaning phase after it performed a right turn on an uncritical cell (\cref{ln_boundary_right,ln_boundary_rear}), which may even happen together with switching to boundary mode (\cref{ln_switchtoboundary1,ln_switchtoboundary2,ln_switchtoboundary3,ln_switchtoboundary4}). Please note, that in these lines, no criticality check is performed. This is done later in the strategy in  \cref{ln_critical_passed}. A cleaning phase is stopped when the agent passes left turns (\cref{ln_boundary_left}) or critical cells (\cref{ln_critical_passed}). Additionally, in boundary mode,  the agent cleans its current position independently from cleaning phases, if it is located within a tail.

\begin{figure}%
\center
\includegraphics[width=4cm]{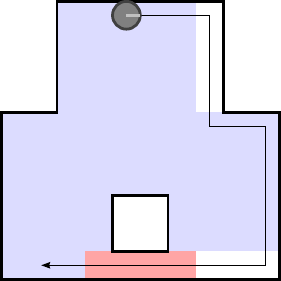}%
\caption{Two agent cleaning phases. The white cells are cleaned as an agent using \emph{SEP} in boundary mode follows the depicted trajectory. The upper phase is stopped when the agent traverses a left turn, the lower one when a critical cell is traversed. Tail cleaning is not depicted here; we will examine what it is good for later.}%
\label{p_cleanstartstop}
\end{figure}

%As was announced in the introduction to this section, our strategy does not clean some boundary cells that are traversed, even though they are not critical. Let us reiterate this behavior will not cause large delays, while at the same time it will prevent the contamination from getting more complex by only local means; as such, it is a key element. 

The above version of the strategy performs a full reset if a spread is recognized (\cref{ln_spreaddetected}). Spreads are recognized if a clean cell located in an agent's perception range becomes contaminated. In case all cells in an agent's perception range have been contaminated before a spread, an agent will be unable to recognize a spread -- however, in this case the agent moves freely through the contamination in search mode, in which spreads are not relevant for its behavior.
%Later, we will use the analysis performed so far and optimize this way of boundary searching. 

%\TODO{Sch�ne Formulierung: The memory requirements for knowledge accumulation impose a strict limitation on the sice of the work areas that can be covered by bounded-memory robots.}

%\TODO{We use the following terminology.}

%\TODO{We use the following convention for}

\section{How Smart Edge Peeling changes contaminations}
\label{s_strategychanges}

We have already seen that $\contaminationset$ is closed with respect to spreads. Thus, it remains to examine $\contaminationset$ is also closed with respect to \emph{SEP} cleaning operations. In order to provide an answer to this question, we will now examine when cleaning phases in a contamination $C \in \contaminationset$ are started and stopped. Then, we will analyze for one single contaminated cell that is being cleaned by an agent, how $\poly{C}$ can possibly be changed. 

In the following lemma and proof, we will show that \emph{SEP} only cleans a cell when it finds itself in one of the eight situations depicted in \cref{p_cleaningcases}.

\begin{figure}%
\center
\includegraphics[width=7cm]{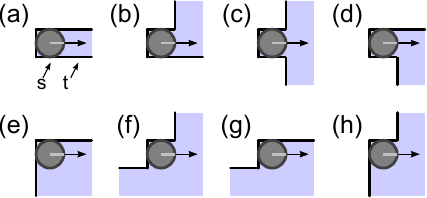}%
\caption{The situations agent cleaning is performed in. A new left turn is only created in situation (e). Situations (a) -- (d) are the four possible tail situations, except for rotational symmetry.}%
\label{p_cleaningcases}%
\end{figure}

\begin{lemma}\label{lem_compltransition_antsweep}
Let $C \in \contaminationset$. Let \emph{SEP} clean a cell of $C$, yielding contamination $D$. Then, $D \in \contaminationset$. 
\end{lemma}

\begin{proof}

Let $C$ consist of more than one cell, otherwise the agent would clean this cell and terminate. Let $s$ be a cell in $C$, let $t$ be a contaminated 4Neighbor of $s$ and let an agent using \emph{SEP} clean $s$ moving to $t$. W.l.o.g. let $t$ be $s$'s east neighbor.
By \cref{s_strategy} the agent starts cleaning following the boundary right hand rule after it turned right in an uncritical boundary cell before. It stops cleaning when traversing critical cells or turning left\footnote{Note: From \cref{alg_smartedgepeeling} one can derive that, if the agent is not located on the last contaminated cell, turns  (block 1, Lines \ref{ln_block1} to \ref{ln_block1end}) are performed before cleaning (block 2, Lines \ref{ln_block2} to \ref{ln_block2end}).}. Also, an agent may clean if located in a tail.

Then, $s$'s north neighbor is clean, as otherwise the agent would not be moving east. Further, $s$'s northwest neighbor is clean, otherwise there would be a U-turn in $\poly{C}$ and hence $C \notin \contaminationset$. As the agent either starts a cleaning phase turning right or continues a cleaning phase or it is located in a tail, $s$'s west neighbor is clean.
It remains to examine the possible contamination states of $s$'s southwest, south, southeast end northeast neighbors, which place constraints on each other.

If $s$'s south neighbor is clean, the southwest one also must be clean, otherwise there would again be a U-turn in $\poly{C}$ and $C \notin \contaminationset$. If the south neighbor is contaminated, the southeast must be contaminated, too, otherwise $s$ would be critical, a contradiction to our assumptions. All the remaining situations are depicted in \cref{p_cleaningcases}; they include all in which the agent is located in a tail. In none of the situations, $\poly{C}$ is deformed in a way it does not consist of four monotonic chains.

Furthermore, the agent does not destroy a contamination's connectivity and does not change the shape of holes. All criteria in \cref{def_contaminationset} are preserved, and $\contaminationset$ is  closed with respect to \emph{SEP} cleaning operations.
\end{proof}

The combination of \cref{lem_compltransition_antsweep,lem1_compltransition_spread} guarantees that across the whole runtime, we never have to deal with other contaminations than the ones in $\contaminationset$. Also, as there can never grow together parts of a contamination enclosing polygon during a spread, no new holes can emerge. We sum up:

\begin{corollary} \label{cor_noconditionschanged}
Let $C \in \contaminationset$ be an initial contamination. Let $D$ be a later contamination resulting of spreads and / or \emph{SEP} cleaning operations on $C$. Then, $D \in \contaminationset$ and no new holes did emerge at any spread that may have happened in between. 
\end{corollary}

We examined earlier how spreads carry out influences on width, height and circumference of contaminations. Now, we need to do the same for agents.  Note that \emph{SEP} cleaning operations never increase a contamination's width and height, as an agent never contaminates cells. Because $\circum{C} = 2\width{C} + 2\height{C} -4$ the circumference is never increased as well. We now examine how they actually get decreased. For this, we need another definition.

\begin{definition}
Let us define an \defi{ear} as a maximal  strip of contaminated cells, each adjacent to the same side of $\boundingbox{C}$.
\end{definition}

\begin{figure}%
\center
\includegraphics[width=45mm]{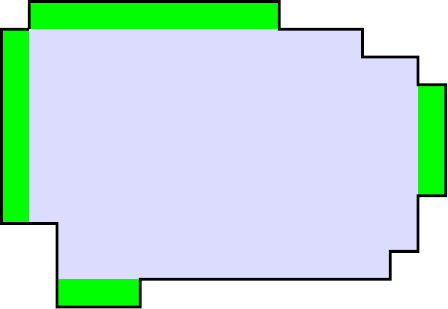}%
\caption{A contamination with ears painted green.}%
\label{p_earintroduction}%
\end{figure}

Ears are depicted in \cref{p_earintroduction}. We use the following naming convention. If an ear is adjacent to the north side of $\boundingbox{C}$, we call it a \defi{north ear}, and so on. Note that in contaminations in $\contaminationset$ there can be no more than one ear per compass direction, because otherwise there would exist an U-turn in $\poly{C}$ between ears touching the same $\boundingbox{C}$ side. Further note, that ears can also overlap, i.e., contaminated cells may belong to more than one ear. For instance, in a contamination consisting of a single cell, the cell marks all four ears. 

\begin{figure*}%
\center
\includegraphics[width=\textwidth]{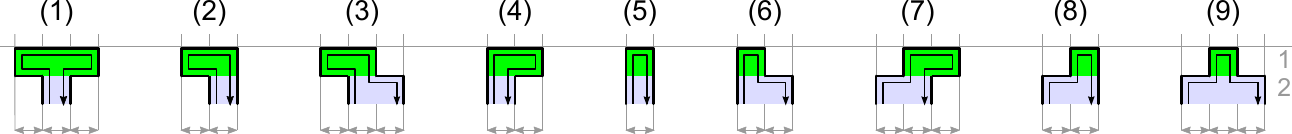}%
\caption{All possible north ear configurations except for stretching. Depicted are the first two layers of the north side of $\boundingbox{C}$ of a simply-connected contamination $C$. The grey, horizontal line marks the north side of $\boundingbox{C}$. The equally colored numbers on the right mark the layer numbers. Ear cells are depicted green. Columns surrounded by grey, vertical lines and marked with a $\leftrightarrow$ at the bottom of the figure can have an arbitrary width of $\geq 1$ units. The arrow trajectories mark how the ears are traversed by agents using \emph{SEP} in boundary mode. Note that while some of the ear configurations might be intuitively considered as impossible in contaminations $\in \contaminationset$, actually all of them can occur. Still, the ears whose first or last turn is a left one do impose constraints on the shape of the rest of the contamination. For example, in contaminations $\in \contaminationset$, type (1) ears can exist; but not all of the four ears can be of type (1); the leftmost and rightmost cells of the ear would be $C$'s west and east ears.}%
\label{p_aabbconfigs}%
\end{figure*}

\begin{lemma}
%[Agent influence on box and circumference of simply-connected contaminations]
\label{lem1_aabbcircumchange_simplyconnected_antsweep}
Let $C \in \contaminationset$, let the outmost hole cell in $C$ be in a layer $\geq 3$. Let $D$ be the contamination after an agent using \emph{SEP} performed one traversal on $C$. Then, $\width{D} \leq \width{C}-2$ and $\height{D} \leq \height{D}-2$, respectively, at least four ears have been cleaned.
\end{lemma}
\begin{proof}
It is easy to assess that $C$ could be cleaned with one agent traversal if $\min(\width{C},\height{C}) \leq 2$. Hence, let us assume that $\min(\width{C},\height{C}) > 2$. 

We will prove that during one traversal, for each compass direction at least one ear is cleaned. W.l.o.g. let us examine the north box side. 

Except for stretchings, there are nine possibilities how an ear can look like. See \cref{p_aabbconfigs} for all nine possible variants of north ears.

% By \cref{lem_simplyconnected_alwayscleanablecells}, we already know that there must be cleanable cells in $C$'s boundary. However, 
%In order to derive how $C$'s width and height are being influenced by agents, we need to have a look on what happens within the outmost layers of cells still in $\boundingbox{C}$.   
In addition to \cref{p_aabbconfigs}, in \cref{p_aabbconfigsexact}, we prove step by step for two example ear configurations that ears are cleaned in one agent traversal. Observe how critical cells within the ears' parts protruding to the east and the west lose their criticality during the passing of the ear so they can get cleaned. The cleaning of other ear variants is performed analogously, and every ear is cleaned when completely passed by an agent in boundary mode. Here we make use of the assumption the outmost hole cell in $C$ is located in a layer $\geq 3$. Otherwise, there could exist holes in layer 2 causing critical cells in layer 1 that would not become uncritical in this way and therefore make the cleaning of an ear impossible.
If an ear is cleaned, either the contamination's width or height is reduced by 1, and its circumference is reduced by at least two. Hence, during one traversal, for each compass direction at least one ear is cleaned, proving the lemma.
\end{proof}

\begin{figure}
\center
\includegraphics[width=7cm]{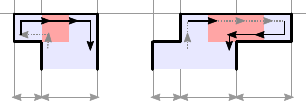}
\caption{Exact \emph{SEP} behavior on two example north ears cleaned after one traversal. The left example is a stretched variant of configuration (2) in  \cref{p_aabbconfigs}, the right one of configuration (7). Cells are marked red by their \emph{initial} state of criticality. In both examples, \emph{SEP} boundary mode trajectories are included.  Every atomic agent movement is marked by an arrow tip on its target. During grey, dotted steps, the agent does not clean. Black, continuous steps mark cleaning of the cell the agent leaves in this step. In both examples, we assume that the agent is not cleaning when entering layer 1. We name the steps alphabetically and reference the line numbers responsible for turning and cleaning in \cref{alg_smartedgepeeling}, where appropriate. \textbf{Left example.} Step A) Agent enters layer 1. B) Moves left (\cref{ln_boundary_left}). C) Turns right twice, cleaning switched on (\cref{ln_boundary_rear,ln_cleaning_normal}), rendering target cell uncritical. D) Moving on (\cref{ln_boundary_front}), still cleaning. E) Right turn (\cref{ln_boundary_right}), still cleaning, leaving layer 1. \textbf{Right example.} Ommiting line numbers known from before. A) Enters layer 1. B) Turns right, cleaning switched on. C) Passing critical cell, cleaning switched off again (\cref{ln_critical_passed}). D) Just moving. E) Double right turn, cleaning switched on, rendering target cell uncritical. F) Still cleaning, again rendering target uncritical. G) Left turn (\cref{ln_boundary_left}). Cleaning switched off. However, we are located in a tail, which is cleaned anyway (\cref{ln_cleaning_tail}). The rest of the configurations can be analyzed analogously.}
\label{p_aabbconfigsexact}
\end{figure}

%Assuming an agent in boundary mode, we painted the respective agent trajectory in every of the 9 configurations. From the trajectories, we can derive that in all of the configurations, the respective ear is cleaned within one agent traversal. Two examples of such analysis closely related to the  \cref{alg_smartedgepeeling} are given in \cref{p_aabbconfigsexact}. As this can be applied to all four sides of the box, we know that any contaminated cell touching the box will be cleaned within one traversal, namely at least four consecutive ears, and the contamination's circumference will be reduced by at least 2 each ear.

By this we also know  $\circum{D} \leq \circum{C}-8$.

\section{More efficient boundary search}
\label{s_quicksearch}

In this section, we will make use of our geometry guarantees and introduce an optimization for boundary searching after a spread in order to resume cleaning earlier after a spread and optimize \emph{SEP}'s runtime. We call this optimization \defi{quick search}. As our optimization only affects the \emph{SEP}'s search mode and not the way of cleaning, the proofs presented so far stay valid.

\begin{lemma}
\label{lem_quicksearch}
Let $C \in \contaminationset$. Let an agent perform \emph{SEP}, be in boundary mode and let a spread happen. After that, the agent can reach the boundary and switch to boundary mode again within three time steps.
\end{lemma}

\begin{proof}
Let $D$ be the outcome of $C$ after the spread. By \cref{cor_noconditionschanged}, $D \in \contaminationset$. W.l.o.g. let the agent be oriented northwards and traverse $C$'s boundary in boundary mode. There are only few possible situations an agent can find itself in when traversing $C$'s boundary left hand rule, right before a spread occurs. They are depicted in \cref{p_quicksearchsituations}. As $\poly{C}$ consists of four monotonic chains, for any of the situations that can occur, green areas are depicted that are guaranteed to be clean in $D$ after the spread occured.

\begin{figure}%
\center
\includegraphics[width=\columnwidth]{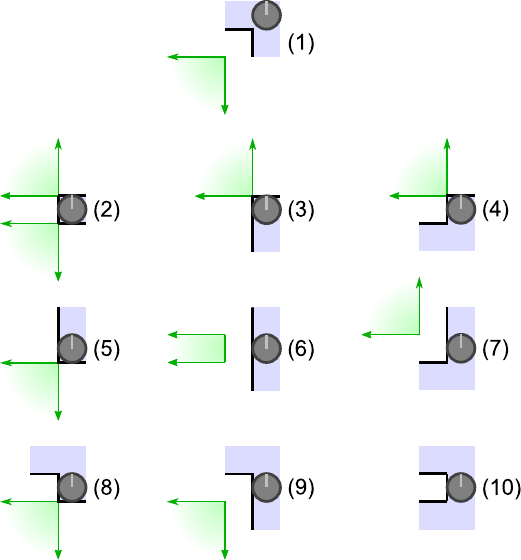}%
\caption{Let $p$ be the agent's position, and let the agent be in boundary mode oriented northwards. If $p$'s west neighbour is contaminated, situation (1) occurs. If $p$'s west neighbour is clean, we can distinct the 9 cases (2) -- (10) dependent on three possible ways $\poly{C}$ may turn at each of the ends of the border edge next to the agent. In every of the situations except for (6) and (10), note the green quadrants are guaranteed to be clean even after a spread because of the monotony of the four chains $\poly{C}$ consists of. In Situation (6), the green stripe is guaranteed to be clean. Situation (10) cannot occur for contaminations in $\contaminationset$.
}%
\label{p_quicksearchsituations}%
\end{figure}

For any of these  possible situations  there are cells in the proximity of the agent clean and not part of a hole after the spread. Hence, we propose the following optimized strategy instead of repeatedly performing a full search  for the boundary, see \cref{p_quicksearchstrategy}: The agent follows a hard-coded path of maximum length three until located at a contaminated cell next to a clean cell. Once located next to one of the depicted cells clean, it turns so that the clean cell is to its left hand side and switches back to boundary mode. If it senses to be located next to a right turn in $\poly{C}$, it also sets the \texttt{lastTurnWasRight} variable accordingly. 
\end{proof}

\begin{figure}%
\center
\includegraphics[width=45mm]{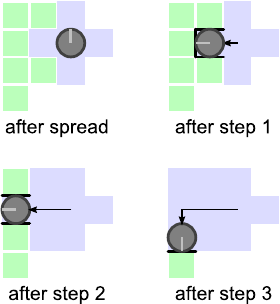}%
\caption{The three hard-coded steps of quick search. After the spread, there is no clean cell within the 4Neighborhood of the agent any more. However, the agent knows one of the green cells must be clean. It then performs the three hard-coded steps depicted until it has found a clean cell, which happens at step 3 at the latest.}%
\label{p_quicksearchstrategy}%
\end{figure}

Spreads add two to both a contamination's width and height. If after a spread an agent manages to clean one ear of every compass direction and another fifth ear, it can reduce both width and height of a contamination by two, and one of them by three, shrinking the contamination's dimensions more than the spread increased them. We now investigate how long this process takes.

%Note, that in contrast to finding the boundary by full search like originally proposed in \cref{s_strategy} after a spread, the agent does not necessarily try to start a cleaning phase right after switching to boundary mode, as it is not necessarily located next to a right turn in $\poly{C}$ when switching to boundary mode. As a consequence, it is important to determine how much time an agent needs to decrease width and height of a contamination more than a spread can compensate for. 

\begin{lemma}
\label{lem_oneandahalftraversalsafterquicksearch}
Let $C \in \contaminationset$. Let an agent be in boundary mode on $C$. Let a spread happen, yielding contamination $D$. Given $d \geq 3\width{C} + 3\height{C} +6$, before the next spread, \emph{SEP} cleaning operations yield a contamination $C' \in \contaminationset$ with $\width{C'} \leq \width{C}$, $\height{C'} \leq \height{C}$ and $\width{C'}+\height{C'} \leq \width{C}+\height{C}-1$.
\end{lemma}

\begin{proof}

Let an agent have traversed $C$'s boundary in boundary mode, and let a spread happen. Let $D \in \contaminationset$ (\cref{cor_noconditionschanged}) be the outcoming contamination. By \cref{lem_4layers} there will be no hole cells in a layer less than 4.

Let us use the following conventions. We denote the most westwards contaminated cell of the north ear as turning point, and analogously for the remaining three compass directions. Each of these turning points is marked by a black dot in \cref{p_ears}, left subfigure of each example. When passing a turning point, an agent turns right, cleaning the ear introduced by the turning point, then leaving it. As the corner cells in $\boundingbox{D}$ cannot be contaminated after the spread (they cannot have had a contaminated 4Neighbor), we know that $D$'s four ears are of type (9) with respect to \cref{p_aabbconfigs}. This ear type does not contain any critical cells, and if traversed, is cleaned without the need of  a cell to be visited twice. 

\begin{figure*}%
\center
\includegraphics[width=\textwidth]{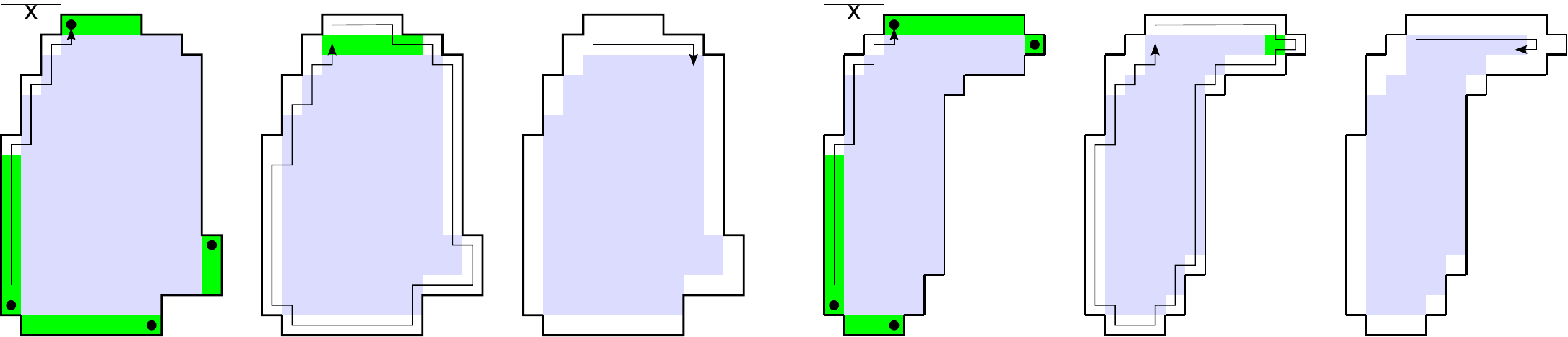}%
\caption{Two examples for agent trajectories when cleaning five ears after quick search. The left three subfigures represent the tree phases described in the proof of \cref{lem_oneandahalftraversalsafterquicksearch}. The right three subfigures describe the same three phases, but on another contamination example. We depict the three phases for two contamination examples to illustrate the two cases occuring at the end of the proof of \cref{lem_oneandahalftraversalsafterquicksearch}. All ears to be cleaned within the next phase are depicted green. Any phase's subfigure contains the agent's trajectory during the phase. In every first phase's subfigure, turning points are marked by dots.}%
\label{p_ears}%
\end{figure*}

We split the time until five ears are cleaned into four phases (each starting at the preceding phase's end or after the spread, respectively):

%\TODO{Neue Phasen: Quick Search, Phase bis zu START CLEANING erstes ear, ab da ein traversal nach lemma da nur einfache ohren, und f�rs letzte ohr m�ssen wir genauer hinsehen.}

%\begin{itemize}
%\item 3 f�r quicksearch
%\item w-3+z bis erster turning point, 
%\item in 2w+2h-4-1 schritten vier ohren vom typ9 weg, der letzte schritt geht drauf weil in der breite ja schon einer verloren wurde. Der Agent befindet sich nun im neuen east ohr, �ber ihm wiederum z freie zellen.
%\item phase 1, until the agent reaches the first turning point,
%\item phase 2, until the first ear is cleaned, and
%\item phase 3, until five ears are cleaned.
%\end{itemize}

\begin{itemize}
\item Phase 1, until the agent reaches the first turning point,
\item phase 2, until the first four ears of type (9) are cleaned,
\item phase 3, until the last ear is cleaned.
\end{itemize}

\emph{Phase 1.} Getting back to boundary mode takes the agent three time steps (\cref{lem_quicksearch}). In the worst case, the agent just missed a turning point. W.l.o.g. let it miss the west one, so the first turning point to pass is the one of the north ear. Between two turning points, an agent moves in a monotonous trajectory (\cref{p_ears}, left subfigure of each example). In the vertical, the agent has to cover a distance of $\height{D}-3$ in the worst case (assuming the north ear started most to the west and its turning point was only missed most closely). In the horizontal, the agent has to cover $x$ cells to reach the westmost cell of the north ear, where $x \geq 1$ (the corner cells in $\boundingbox{D}$ cannot be contaminated for they cannot have had a contaminated adjacent cell). Overall, phase 1 needs $\height{D} + x$ time steps.

\emph{Phase 2.} As a spread just happened, the outmost hole cell in $D$ can be only in a layer $\geq 4$ \cref{lem_4layers}. By \cref{lem1_aabbcircumchange_simplyconnected_antsweep}, within one traversal the agent is able to clean $D$'s north, east, south and west ear. One traversal takes $2\width{D}+2\height{D}-4$ time steps (\cref{lem_circumferencelenght}, \cref{def_traversal}) and is depicted in \cref{p_ears}, middle subfigure of each example. As by its cleaning operations, the contamination lost one unit of height in the meantime, the agent needs even one step fewer than a traversal: $2\width{D}+2\height{D}-5$ time steps. After that, the agent is located within a new north ear in a horizontal distance of $x$ to the north side of the original $\boundingbox{D}$.

\emph{Phase 3.} With respect to the original $\boundingbox{D}$ the Agent is located in layer 2 and hole cells can only exist in layers $\geq 4$, so holes cannot cause critical cells in the north ear to clean. Additionally, by the traversal performed in phase two, all cells adjacent to the east side of  $\boundingbox{D}$ are clean.  Hence, the agent needs at most $\width{D}-x-2$ cells to reach the east end of the north ear. There are two cases:
\begin{itemize}
\item The ear does not contain any cells protruding to the east, namely has been of types (2), (3), (5), (6), (8) or (9) with respect to \cref{p_aabbconfigs}. In this case, the agent cleans the ear's last cell and heads south.
\item The ear does contain cells protruding to the east, it has been of types (1), (4), or (7). In this case, the agent finds itself in the eastmost cell of the north ear, which however is also an east ear. It cleans the cell, as it is also a tail, and heads west again. In this case, contrary to our expectations, the agent cleaned an east ear, not a north one. 
\end{itemize}
Both cases consume one further time step. Phase 3 needs $\width{D}-x-1$ time steps. A contamination example yielding the former case is depicted in the three left subfigures of \cref{p_ears}, one yielding the latter case in the three right subfigures. Phase 3 is depicted in the right subfigure of each example.

All three phases together need $3\width{D} + 3\height{D}-6$ time steps, which, by Obs. \ref{observation_spread}, equals $3\width{C} + 3\height{C}+6$ time steps. Let the contamination after this period of time be $C'$. By \ref{cor_noconditionschanged}, $C' \in \contaminationset$. Furthermore, for every compass direction, one ear has been cleaned, and one additional ear has been cleaned, so $\width{C'} \leq \width{C}$ and $\height{C'} \leq \height{C}$. By the fifth ear cleaned, additionally $\width{C'}+\height{C'} \leq \width{C}+\height{C}-1$.
\end{proof}

\section{Correctness and run time} \label{s_correctnessruntime}

We now prove the theorem already stated in the introduction. Let $\lambda$ denote the maximum length of all shorter edges of the rectilinear holes inside a contamination $C \in \contaminationset$ ($\lambda = 0$ if there do not exist such). First, let us recall the theorem. 

\begin{theorem}\label{main-theo}
Given speed $d \geq 3(h+w) +6$, and starting from a contaminated cell, strategy \emph{SEP} cleans each contamination in $\contaminationset$ of height $h$ and width $w$ in at most $(\frac{\lambda}{2} + h + w +5) d$ many steps.
\end{theorem}

% Worst case: Maximales loch.

\begin{proof}
We use the following naming convention: $C_i$ is the contamination that evolved out of $C$ by agent cleaning operations and spreads until the end of time step $i$. As the initial contamination $C$ is in $\contaminationset$, all $C_i$ are as well (\cref{cor_noconditionschanged}), so all the below referenced lemmata are applicable. 

During the first spread phase, $d$ is large enough to allow the agent to find the boundary (\cref{lem_searchtime}) and perform at least one full traversal (\cref{def_traversal}). In the worst case, the agent is unable to reduce the contamination's dimensions due to badly placed holes. In this case, the agent has to wait for the first spread, yielding contamination $C_d$ with height $h+2$ and width $w+2$ (Obs. \ref{observation_spread}). By $d \geq 3(h+w) +6$ and \cref{lem_oneandahalftraversalsafterquicksearch} we know that after the spread, the agent decreases the contamination's width and height more than the spread did increase them. Hence, $\width{C_{2d-1}} + \height{C_{2d-1}} \leq \width{C_{d-1}} + \height{C_{d-1}}-1$. 

This reasoning is applicable from any further spreadphase's end to the next: $\width{C_{(i+1)d-1}} + \height{C_{(i+1)d-1}} \leq \width{C_{id-1}} + \height{C_{id-1}}-1$. From the end of any spreadphase to the end of the next. Overall, the agent needs at most $w+h+4+1$ spread phases to completely clean the contamination. 

 Greater $d$ allow for more width and height reduction per spread phase,  leading to fewer needed spread phases. Holes however may impair the agent's usage of such large $d$ and force it to wait for further spreads. After $\frac{\lambda}{2}$ spread phases, all holes are fully contaminated, leading to $\frac{\lambda}{2} + h + w + 5$ necessary spread phases overall. 
\end{proof}

%By \cref{lem_4layers} after any spread, the outmost hole cell in a contamination can only be located in layer $\geq 4$. Thus, given $d \geq 5(w+h)-3$ there is enough time within a spread phase to find the boundary by quick search (\cref{lem_quicksearch}) and apply \cref{lem1_aabbcircumchange_simplyconnected_antsweep} twice. For such $d$, an agent can clean at least 8 ears, which allows for a guaranteed width and height reduction of 2 per spread phase. Hence,
%\begin{corollary}
%Given $d \geq 5(w+h)-3$, analogously to the proof of \cref{main-theo}, we get a cleaning time of $\left(\frac{\lambda}{2} + \frac{\min(h,w)}{2}+2\right)d$.
%\end{corollary}

Without holes or holes located in deeper layers, the agent can make use of even larger $d$, leading to an arbitrarily large reduction of the contamination's width and height per spread phase and therefore fewer necessary spread phases. Hence we can conclude that while our strategy was designed for purely local handling of more complex scenarios, it also competes well on simply-connected and static scenarios.

\section{Lower bounds}         \label{s_lowerbounds}
\label{s_connectivitymaintenance}
Our lower bounds are based on the following isoperimetric inequality that can be found, e.g., in Altshuler et al. (\cite[Theorem 8]{altshulerdynamic}).

\begin{theorem}     \label{diamond-theo}
Let $C$ be a contamination of $c$ cells. Then at least $2 \sqrt{2c-1}$ new cells will be contaminated in the next spread.
\end{theorem}

This bound is attained for the diamond shapes (or $L_1$-circles) that result from the spreading of a single contaminated cell; see \cref{p_diamonds}. Here all but four newly contaminated cells are infected by {\em two} neighbors, minimizing the contamination increase.

\begin{theorem}         \label{lowbound-theo}
An $h \times h$ square $C$ cannot be cleaned at speed $d <\sqrt{2}\  2h - 4$.
\end{theorem}
\begin{proof}
Before the first spread occurs, at least $c=h^2-d$ cells of square $C$ are still contaminated. By Theorem~\ref{diamond-theo}, at least $2 \sqrt{2(h^2-d)-1}$ cells will become newly contaminated. If this number is $>d$, an even larger number $c' >c$ of  cells will remain contaminated before the second spread occurs, and so on, proving that cleaning $C$ is impossible. Because of
\[
 2 \sqrt{2(h^2-d)-1}    >  d  \ \Longleftrightarrow \ 8 h^2 +12 > (d+4)^2
\]
the claim follows from $d+4 < \sqrt{2}\  2h$.
\end{proof}

Since the lower bound of Theorem~\ref{lowbound-theo} is not based on the robot's incomplete knowledge it applies to optimum offline solutions, too. The next result, in contradistinction, holds only for the online strategies we are considering.

\begin{theorem}     \label{neogreedy-theo}
Let $G$ be a strategy that always cleans the current cell if contaminated, moves to a contaminated cell in its 8-neighborhood, if there is one, and rests, otherwise. Then $G$ cannot clean all strips of length $w$ at speed $d < 4 (w+h) -16$.
\end{theorem}
\begin{proof}
Consider a single row of cells, as shown in \cref{p_uncleanableclass}. The robot starts from an interior cell, cleans it and moves straight to the left or to the right until the end of the strip is reached. Let us assume it moves to the right. At this point we define the initial contamination such that only one of $w=l+2$ cells is situated to the left of the robot's start position, which remains contaminated as the robot cleans the other $l+1$ cells. While the robot rests at the rightmost cell, contamination spreads from the leftmost cell as shown in \cref{p_uncleanableclass}. After $l$ spreads, a diamond shape of $2 l^2 + 2 l +1$ cells is contaminated, among them the cell to the left of the robot. Before the $(l+1)$-st spread occurs, $c=2 l^2 + 2 l +1 - d$ cells are left contaminated. By \cref{diamond-theo}, at least $2 \sqrt{2c-1}$ cells will become newly infected. We have
\begin{align*}
 &&2 \sqrt{2(2 l^2 + 2 l  -d +1)-1}    &>  d  \\ 
& \Longleftrightarrow & 16 l^2 +16 l + 20 &> (d+4)^2
\end{align*}
which holds true because of $d < 4 l -4$. Hence, the increase in contamination will always exceed the maximum number $d$ of cells the robot can clean between spreads.
\end{proof}
\begin{figure}[tbp]
\center
\includegraphics[width=\columnwidth]{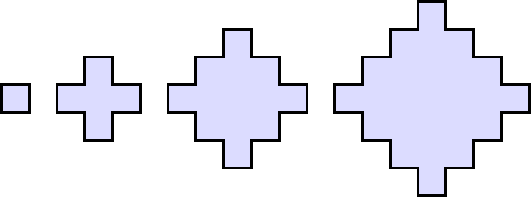}%
\caption{From the left to the right: A minimal contamination after 0, 1, 2 and 3 spreads.}%
\label{p_diamonds}
\end{figure} 
\begin{figure}[tbp]
\center
\includegraphics[width=5cm]{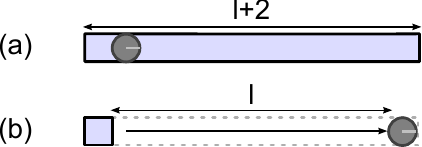}%
\caption{\textbf{(a)} A straight contaminated line of cells of length $l+2$ with the greedy agent's starting position. \textbf{(b)}  The greedy agent destroying it's connectivity. }%
\label{p_uncleanableclass}
\end{figure} 

The same result can be shown if we allow a greedy strategy $G$ to perform a kind of search for contaminated cells once no contaminated cell is left in its current neighborhood. This is because the robot is a finite state machine, so that only a cyclic search path pattern of constant diameter could result from this capability. If the start and end positions of the pattern are not equal, the agent translates through the space in a constant direction, never visiting cells on the opposite direction.

\section{Conclusions and further research}
\label{s_conclusions}

In this article, we presented a cleaning strategy \emph{SEP} enabling a single finite automaton robot to clean expanding contaminations by only local means. \emph{SEP} maintains geometric invariants and additionally ensures that the contaminated cells stay connected. Furthermore, we proved that greedy strategies violating the latter principle need greater spreading times $d$ than \emph{SEP} in general. We considered contaminations $\in \contaminationset$, i.e., with certain limitations on their geometric complexity. 
Besides improving lower bounds, our results suggest two main directions to obtain qualitative enhancements on the task of cleaning expanding contaminations.

One way of generalizing our work is to consider contaminations with arbitrarily complex shapes (\cref{p_complicated}), which inadvertently raise further challenges. For example, new holes can emerge in spreads and be of likewise geometrical complexity, or existing holes may split. Some of the lemmata we in this article can already be generalized to higher geometrical complexities. However, to be able to generalize the entire work, more geometrical analysis is necessary. 
\begin{figure}[tbp]
\center
\includegraphics[width=\columnwidth]{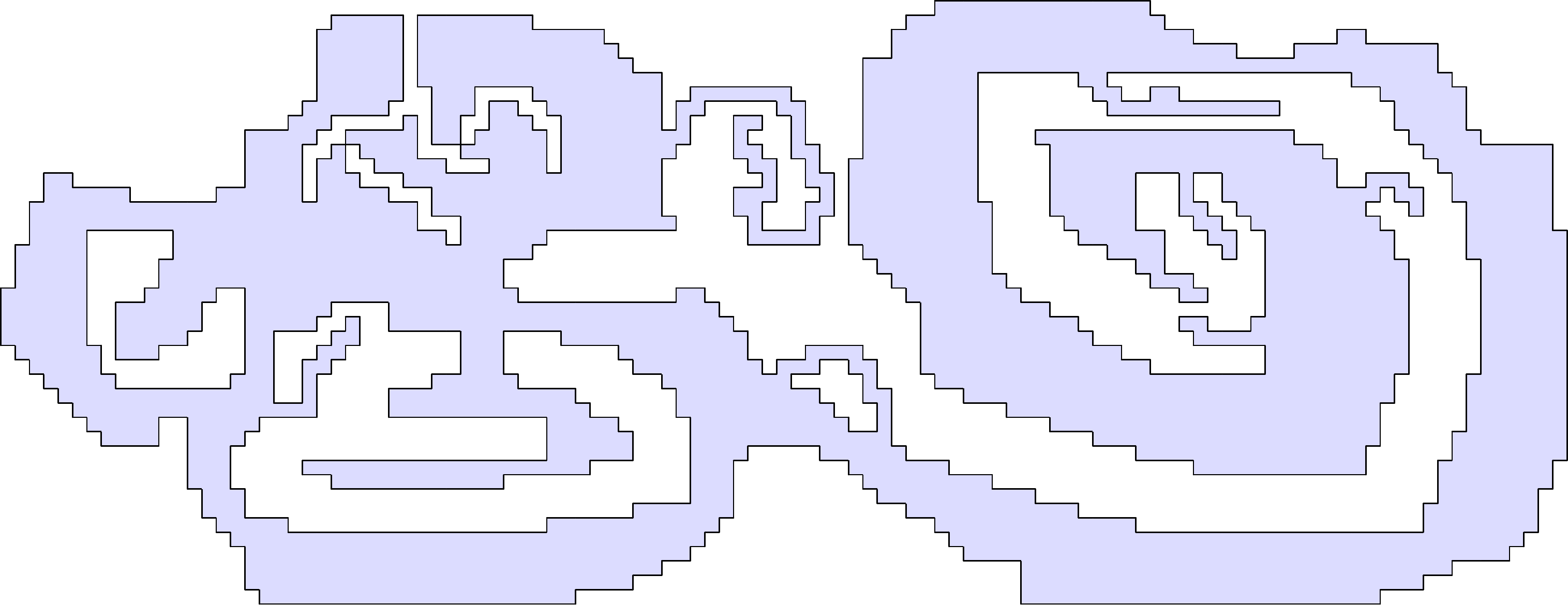}%
\caption{A contamination with high geometrical complexity. In a spread, new holes would emerge and existing holes would be split.}%
\label{p_complicated}
\end{figure}

A further interesting question is how to use a swarm of $k$ agents cleaning expanding contaminations in parallel in order to increase cleaning speed and exhibit fault tolerance known from biological systems.

We are confident that both ways of generalization lead to qualitatively new results. They are subject to our current research.

\bibliographystyle{plain}
\bibliography{biblio}

\begin{thebibliography}{10}

\bibitem{altshulercomplexities}
Y.~Altshuler and A.~M. Bruckstein.
\newblock Static and expanding grid coverage with ant robots: Complexity
  results.
\newblock {\em Theoretical Computer Science}, 412(35):4661--4674, 2011.

\bibitem{altshuler2005swarm}
Y.~Altshuler, A.~M. Bruckstein, and I.~A. Wagner.
\newblock Swarm robotics for a dynamic cleaning problem.
\newblock In {\em Swarm Intelligence Symposium, 2005. SIS 2005. Proceedings
  2005 IEEE}, pages 209--216. IEEE, 2005.

\bibitem{altshulershape}
Y.~Altshuler, I.~A. Wagner, and A.~M. Bruckstein.
\newblock Shape factors effect on a dynamic cleaners swarm.
\newblock In {\em ICINCO (MARS Workshop)}, 2006.

\bibitem{altshulerdynamic}
Y.~Altshuler, V.~Yanovski, I.A. Wagner, and A.M. Bruckstein.
\newblock Multi-agent cooperative cleaning of expanding domains.
\newblock {\em The International Journal of Robotics Research},
  30(8):1037--1071, 2011.

\bibitem{lawnmowingmilling}
E.~M. Arkin, S.~P. Fekete, and J.~S.~B. Mitchell.
\newblock Approximation algorithms for lawn mowing and milling.
\newblock {\em Computational Geometry}, 17(1):25--50, 2000.

\bibitem{lionsberger}
F.~Berger, A.~Gilbers, A.~Gr{\"u}ne, and R.~Klein.
\newblock How many lions are needed to clear a grid?
\newblock {\em Algorithms}, 2(3):1069--1086, 2009.

\bibitem{choset2001coverage}
H.~Choset.
\newblock Coverage for robotics--a survey of recent results.
\newblock {\em Annals of mathematics and artificial intelligence},
  31(1-4):113--126, 2001.

\bibitem{lionsdumitrescu}
A.~Dumitrescu, I.~Suzuki, and P.~Zylinski.
\newblock Offline variants of the lion and man problem.
\newblock In {\em Proceedings of the twenty-third annual symposium on
  Computational geometry}, pages 102--111. ACM, 2007.

\bibitem{fomin2008annotated}
F.~Fomin and D.~Thilikos.
\newblock An annotated bibliography on guaranteed graph searching.
\newblock {\em Theoretical Computer Science}, 399(3):236--245, 2008.

\bibitem{gabriely2002spiral}
Y.~Gabriely and E.~Rimon.
\newblock Spiral-stc: An on-line coverage algorithm of grid environments by a
  mobile robot.
\newblock In {\em Proceedings. ICRA'02. IEEE International Conference on
  Robotics and Automation}, volume~1, pages 954--960. IEEE, 2002.

\bibitem{gabriely}
Y.~Gabriely and E.~Rimon.
\newblock Competitive on-line coverage of grid environments by a mobile robot.
\newblock {\em Computational Geometry}, 24(3):197--224, 2003.

\bibitem{galceran2013survey}
E.~Galceran and M.~Carreras.
\newblock A survey on coverage path planning for robotics.
\newblock {\em Robotics and Autonomous Systems}, 61(12):1258--1276, 2013.

\bibitem{garnier2007biological}
S.~Garnier, J.~Gautrais, and G.~Theraulaz.
\newblock The biological principles of swarm intelligence.
\newblock {\em Swarm Intelligence}, 1(1):3--31, 2007.

\bibitem{gonzalez2005bsa}
E.~Gonzalez, O.~Alvarez, Y.~Diaz, C.~Parra, and C.~Bustacara.
\newblock Bsa: a complete coverage algorithm.
\newblock In {\em Robotics and Automation, 2005. ICRA 2005.}, pages 2040--2044.
  IEEE, 2005.

\bibitem{henrich}
D.~Henrich.
\newblock Space-efficient region filling in raster graphics.
\newblock {\em The Visual Computer}, 10(4):205--215, 1994.

\bibitem{elmar}
C.~Icking, T.~Kamphans, R.~Klein, and E.~Langetepe.
\newblock Exploring simple grid polygons.
\newblock {\em Computing and Combinatorics}, pages 524--533, 2005.

\bibitem{6269300}
A.~Ntawumenyikizaba, Hoang~Huu Viet, and TaeChoong Chung.
\newblock An online complete coverage algorithm for cleaning robots based on
  boustrophedon motions and a* search.
\newblock In {\em 2012 8th International Conference on Information Science and
  Digital Content Technology (ICIDT)}, volume~2, pages 401--405, 2012.

\bibitem{wagner1999distributed}
I.~A. Wagner, M.~Lindenbaum, and A.~Bruckstein.
\newblock Distributed covering by ant-robots using evaporating traces.
\newblock {\em Robotics and Automation, IEEE Transactions on}, 15(5):918--933,
  1999.

\bibitem{altshulerstatic}
I.A. Wagner, Y.~Altshuler, V.~Yanovski, and A.M. Bruckstein.
\newblock Cooperative cleaners: A study in ant robotics.
\newblock {\em The International Journal of Robotics Research}, 27(1):127--151,
  2008.

\bibitem{zelinsky1993planning}
A.~Zelinsky, R.~Jarvis, J.~Byrne, and S.~Yuta.
\newblock Planning paths of complete coverage of an unstructured environment by
  a mobile robot.
\newblock In {\em Proceedings of international conference on advanced
  robotics}, volume~13, pages 533--538, 1993.

\end{thebibliography}

\end{document}